\documentclass[12pt]{amsart}

\usepackage{amsmath,amsthm, amsfonts, fullpage}
\usepackage{graphicx, rotating}

\newtheorem{thm}{THEOREM}
\newtheorem*{thm*}{THEOREM}
\newtheorem{cor}{Corollary}
\newtheorem{lem}{Lemma}
\newtheorem{prop}{Proposition}
\theoremstyle{definition}

\theoremstyle{remark}
\newtheorem{rem}{Remark}
\newtheorem{eg}{Example}

\newcommand{\bx}{\mathbf{x}}

\newcommand{\beq}{\begin{equation}}
\newcommand{\eeq}{\end{equation}}
\newcommand{\bseq}{\begin{subequation}}
\newcommand{\eseq}{\end{subequation}}

\newcommand{\p}{\partial}

\newcommand{\la}{\lambda}

\newcommand{\si}{\sigma}

\newcommand{\bna}{\begin{eqnarray}}
\newcommand{\ena}{\end{eqnarray}}
\newcommand{\bea}{\begin{eqnarray*}}
\newcommand{\eea}{\end{eqnarray*}}
\newcommand{\ben}{\begin{enumerate}}
\newcommand{\een}{\end{enumerate}}
\newcommand{\bi}{\begin{itemize}}
\newcommand{\ei}{\end{itemize}}

\begin{document}

\title{Remarks on combinatorial aspects of the KP Equation}


\author{Shabnam Beheshti}
\address{Department of Mathematics, Rutgers, The State University of New Jersey, Piscataway, NJ 08854}
\email{beheshti@math.rutgers.edu}

\author{Amanda Redlich}
\address{Department of Mathematics, Bowdoin College, Brunswick, ME, 04011}
\email{aredlich@bowdoin.edu}

\begin{abstract}
We survey several results connecting combinatorics and Wronskian solutions of the KP equation, contextualizing the successes of a recent approach introduced by Kodama, et.~al.  We include the necessary combinatorial and analytical background to present a formula for generalized KP solitons, compute several explicit examples, and indicate how such a perspective could be used to extend previous research relating line-soliton solutions of the KP equation with Grassmannians.

\vspace{2mm}

\noindent Keywords: Integrable PDE, Combinatorics, Kadomtsev-Petviashvili, Hirota Derivative, Wronskian, Faa di Bruno, Vandermonde matrix
\end{abstract}

\maketitle



\section{Introduction}
Integrable nonlinear partial differential equations (PDEs) have been a source of investigation for a vast array of physical systems:  shallow-water waves modeled by the Korteweg-de Vries (KdV), Kadomtsev-Petviashvili (KP) and Boussinesq Equations \cite{GGKM, KP, Bouss}, condensed matter theory and optics studied by way of the sine-Gordon and cubic nonlinear Schr\"odinger (NLS) equations \cite{SG, AKNS-2, ZS-I}, gravitation and electrodynamics described by Yang-Mills, Einstein vacuum and Einstein-Maxwell equations \cite{ZB-I, Ale80, GurJanXan}.  In general, nonlinear PDEs (or systems of PDEs) are extremely difficult to solve explicitly; however, a variety of techniques exist to analyze those which are integrable.  One key feature of an integrable PDE is the existence of a Lax pair, or Lax system, i.e., an overdetermined linear system for which the equation of interest arises as a compatibility condition \cite{Lax68}.  In many instances, this associated system gives rise to algorithmic solution-generating mechanisms (e.g., B\"acklund transforms, Inverse Scattering Mechanism).

First used to model ion waves in plasmas under weak transverse perturbations, the Kadomtsev-Petviashvili Equation, 
\beq\label{eq:KPu}
(-4u_t +6uu_x+u_{xxx})_x + 3\si u_{yy} =0,
\eeq
is considered a prototypical integrable evolution equation in $(2+1)$-dimensions \cite{KP}.  Here, $u=u(x,y, t)$ and subscripts denote partial derivatives with respect to the variable indicated 
and $\si = \pm 1$.  For the subsequent discussion, we shall consider only the case $\si=+1$, known as the stable or KP-II Equation.  A natural generalization of the well-studied integrable KdV Equation, modeling shallow water waves, Equation \eqref{eq:KPu} has been used to investigate the dynamics of Bose-Einstein Condensates (BECs) \cite{Tsuchiya} and various gravitational models (Ward's Chiral model, 2+1 dilaton theories) \cite{Dimakis}, among other physical models of significance.  

More recently, connections between the KP hierarchy and combinatorics have been investigated in several contexts: Young tableaux have been used to analyze the behavior of solutions to the KdV equation \cite{Deift}, chord diagrams allow for classification of line-soliton solutions to the stable KP Equation \cite{CK-1, CK-3}, Topelitz determinants from random matrix theory have played a role in the analysis of certain commuting differential operators \cite{ItsTracyWidom, ItsWidomTracy}.  This novel perspective has led to significant progress in understanding classically integrable PDE; on the combinatorics side, the cross-fertilization has created multiple new areas of research.

The main purpose of this communication is to survey a recent successful approach for analyzing various soliton solutions of the KP equation using combinatorics appearing in the body of work \cite{CK-1, CK-3, KodamaPierce, KW-1, KW-2, CK-2, KW-3}; the general method is to generate a family of soliton solutions to the KP equation with a certain combinatorial structure, and then analyze the KP solutions produced in light of this underlying structure.  This gives a bijection between combinatorial objects (e.g. chord diagrams, plabic graphs, Grassmann necklaces) and features of the associated soliton solutions in large and small scales (e.g. contour plots, resonant graphs).  The bijection informs and is influenced by the family of solutions studied.

It is our intention to introduce the basic vocabulary and quantities necessary to understanding the nature of these problems to newcomers in both nonlinear analysis and enumerative combinatorics.  We also indicate a perspective by which several of the more technical formulas, better known in the integrable PDEs literature, may be used to extend current results relating line-solitons to Grassmannians.  Consequently, the rest of the paper is organized as follows.  First, we discuss integrability of the KP-II equation in terms of Hirota derivatives and determinantal solutions, identifying combinatorial features throughout the discussion.  Next, a two part discussion is devoted to establishing a combinatorial framework for generalized Wronskian solutions of the KP Equation.  We then indicate how to interpret the formulas in the setting of Grassmannians and calculate several explicit examples.  Finally, in light of these facts, we suggest a general program of investigation for future work.

\section{Integrability of the KP Equation}
Various aspects of \eqref{eq:KPu} have been studied using techniques from nonlinear analysis, algebraic geometry, and integrable systems, among other fields.  It admits a Riemann-Hilbert formulation \cite{AblowitzClarkson}, as well as successful use of the Inverse Scattering Mechanism (ISM) \cite{FokasAblowitz} and the Hirota method \cite{Hirota, Satsuma}; local and global well-posedness of the Cauchy Problem are established \cite{Bourgain, Tzvetkov}; KP Hierarchies may be interpreted in the context of algebraic geometry as flows on moduli spaces \cite{Krichever, BenZvi}.

Since we are following the main lines of enquiry in \cite{CK-1, CK-2, CK-3, KW-1, KW-2}, the remainder of this paper concerns Hirota integrability of the KP-II Equation, and in particular, combinatorics of certain Wronskian solutions.  We will restrict our attention to this family of solutions to \eqref{eq:KPu} and refer the reader to \cite{Hirota} for a careful discussion of other determinantal methods used in the study of bilinearizable PDE (e.g.,  Grammians, Pfaffians, Casoratians, etc.).

Recall the Hirota derivatives to be binary differential operators acting on smooth functions $a(x), b(x)$ of one real variable, defined by
\begin{eqnarray*}
D_x^n(a,b) &=& \left( \frac{\partial}{\partial x} - \frac{\partial}{\partial w}   \right)^n a(x) b(w) \Big{|}_{w=x} = \frac{\partial^n }{\partial w^n}a(x+w)b(x-w) \Big{|}_{w=0} \\
D_t^mD_x^n(a,b) &=& \frac{\partial^m }{\partial s^m} \frac{\partial^n }{\partial w^n} a(t+s,x+w)b(t-s,x-w)\Big{|}_{x=0,w=0},
\end{eqnarray*}
where $m,n = 0, 1, 2, \ldots$.  An equation is said to be bilinearizable if it can be recast, under an appropriate transformation, as an equation in terms of Hirota derivatives.  Examples of bilinearizable PDEs include the KdV, KP, Sawada-Kotera, Boussinesq, and Toda lattice equations, to name a few.
\begin{eg}
Under the logarithmic tranformation 
\beq\label{eq:logtransf}
u(x,y,t) = 2 (\log \tau(x,y,t))_{xx},
\eeq
the original KP equation \eqref{eq:KPu} is transformed to the following Hirota or bilinear form ($\si=1$):
\begin{eqnarray}\label{eq:KPtau}
0 &=& \left[ D_x ^4-4D_xD_t+3D_y^2 \right] ( \tau, \tau) \\
&=& (\tau\tau_{xxxx}-4\tau_x\tau_{xxx}+3\tau_{xx}^2)-4(\tau\tau_{xt}-\tau_x\tau_t)+3(\tau\tau_{yy}-\tau_y^2).\nonumber
\end{eqnarray}
\end{eg}
\begin{rem}\label{rem:fourier}
Motivated by \eqref{eq:logtransf}, observe that if  $\tau_x = u \tau$, $\tau_y = \tau_{xx}$, $\tau_t = \tau_{xxx}$, then $u = 2(\log \tau)_{xx}$ is a solution to \eqref{eq:KPu}.  Since $\tau$ satisfies linear equations in $y$ and $t$, a general solution in terms of a Fourier transform can be found to be
$$
\tau(x,y,t) = \int_\mathcal{C} e^{kx+k^2y+k^3t} d\mu(k).
$$
Here $d\mu(k)$ is an appropriate measure on the complex curve of integration $\mathcal{C}\in \mathbb{C}$.  On the other hand, a finite-dimensional solution may be considered by using a (finite) Fourier transform having point measure ${\displaystyle d\mu(k) = \sum_{j=1}^N \rho_j \delta(k-k_j)dk}$, $\rho_j, \,k_j \in \mathbb{R}$.  In this case, $\tau$ is given instead by the sum
$$
\tau(x,y,t)= \sum_{j=1}^N \rho_j E_j(x,y,t),
$$
with $E_j(x,y,t)= e^{k_jx + k_j^2 y+k_j^3 t}$.  It is this formulation from which the authors of \cite{CK-1, KW-1, KW-2} are able to extract interesting combinatorial structure and which we will pursue further.
\end{rem}

\subsection{Wronskian integrability}
Those equations which can be written in Hirota form often admit determinantal solutions.  A PDE (or system of PDEs) is said to be Wronskian-integrable if it can be cast as a solvable system of equations in the entries of a particular Wronskian matrix; see ~\cite{Hirota, Satsuma} and references therein.   In \cite{FN-1, FN-2, FN-3, MaComplexitons}, the authors successfully generate solutions of the KdV, KP and Boussinesq Equations using the Wronskian 
\beq\label{eq:tauWronsk}
\tau= Wr(\phi_1, \phi_2, \ldots , \phi_n) = \left| \begin{array}{ccc} \phi_1 & \cdots & \phi_n \\ {\phi_1}^\prime & \cdots & {\phi_n}^\prime \\ {\phi_1}^{\prime\prime} & \cdots & {\phi_n}^{\prime\prime} \\ \vdots &  & \vdots \\ {\phi_1}^{(n-1)} & \cdots & {\phi_n}^{(n-1)} \end{array} \right| ,
\eeq
where $ \phi = \phi(x_1, x_2, x_3) = \phi(x,y,t)$ and ${\displaystyle {\phi_j}^{(i)}=\frac{\partial^i \phi_j}{\partial x^i} }$, for $i=1,2,3$ and $j=1, 2, \ldots n$.  In this manner, a transformed PDE in $\tau$ can be re-expressed as a system of PDEs in the generators\footnote{Note that this terminology is not to be confused with the generating functions arising in combinatorics.} $\phi_j$,  $j=1, 2, \ldots , n$.  When the PDE under consideration involves $N$ variables, it is possible to view the $\phi_j$ in \eqref{eq:tauWronsk} as functions of $\bx=(x_1 , \ldots , x_N)$.  In fact, such generalized $\phi_j$ may be used to solve PDEs with fewer independent variables by imposing additional constraints.

\begin{rem}\label{rem:schur}
It turns out that in the case of the KP Equation, the algebraic relation in the derivatives of $\tau$ appearing in \eqref{eq:KPtau} can be realized, in fact, as a Pl\"ucker relation on a Grassmannian.  Let $0 \leq  I_0 < I_1 < \cdots < I_{n-1}$ and define
$$
[I_0, I_1, \ldots , I_{n-1}] =  \left| \begin{array}{ccc} \phi_1^{(I_0)} & \cdots & \phi_n^{(I_0)} \\  \vdots &  & \vdots \\ {\phi_1}^{(I_{n-1})} & \cdots & {\phi_n}^{(I_{n-1})} \end{array} \right| .
$$
Then one may view $[I_0 , I_1, \ldots , I_{n-1}]$ as the algebraic symbol representing the determinant of $n$ vectors in $\mathbb{R}^{\infty}$ and the well-known Pl\"ucker Relation is given by
$$
\sum_{j=0}^n (-1)^j [\alpha_0 , \ldots ,  \alpha_{j-2}, \beta_j] [ \beta_0 , \ldots , \hat{\beta}_j , \ldots , \beta_n].
$$
In this notation, $\tau = Wr(\phi_1, \phi_2, \ldots , \phi_n)=[0,1, \ldots, n-1]$, where the numbers $0,1, \ldots, n-1$ represent derivatives of the functions $\phi_j$ in the determinant appearing in \eqref{eq:tauWronsk}, and the associated Pl\"ucker Relation for $\alpha = \{0,1, \ldots , n-2\}$ and $\beta=\{0,1, \ldots , n-3 , n-1 , n , n+1\}$ recovers the KP Equation.  One proof of this result connects the symbol $[I_0, I_1, \ldots , I_{n-1}]$ with the derivatives of the $\tau$ function using Schur polynomials and re-expresses  \eqref{eq:KPtau} in terms of Young diagrams \cite{Hirota, MJD}.   More generally, the Sato theory establishes that all of the equations of the KP Hierarchy, when viewed in terms of appropriate $\tau$ functions, can be realized as Pl\"ucker coordinates of the Grassmannian $Gr(n, \infty)$. 
Explicit treatment of these statements is found in \cite{Sato}, for instance.
\end{rem}

Let us return to \eqref{eq:tauWronsk} in the KP setting and assume $\phi_j(\bx)= \phi_{j}(x_1 , \ldots , x_N)$.  Then the additional constraint imposed on the $\phi_{j}(x_1 , \ldots , x_N)$ is that they satisfy a linear system
\begin{equation}\label{eq:hir}
\frac{\partial \phi_j}{\partial x_m}=\frac{\partial^{m}\phi_j}{\partial x_1^m}, \qquad m= 1 , \ldots , N.
\end{equation}
In light of Remark \ref{rem:fourier}, we make the following immediate observation: any function which may be described as the result of taking the derivative of $e^{px_1+p^2x_2+\cdots +p^N x_N}$ with respect to $p$ arbitrarily many times will also solve the KP equation in a Wronskian formulation.  This statement, also appearing as a remark in \cite{Hirota}, is stated formally below, and will be used as a starting point for our subsequent generalizations.

\begin{lem}\label{lem:explem}  Let $\phi (p,\bx) = e^{f(p,\bx)}$, where $f(p,\bx) = px_1 +p^2x_2 + p^3x_3 + \cdots +p^N x_N$.  Then for each $p$, $\phi(p,\bx)$ satisfies \eqref{eq:hir}. 
Furthermore, any collection of ${\displaystyle \frac{\partial^i \phi_j}{\partial p^i}}$ generate solutions to \eqref{eq:KPtau}, $i \in \mathbb{N}$.
\end{lem}
Since the equations in \eqref{eq:hir} are linear, a straightforward induction on $m$ proves the first statement of the lemma.  Furthermore, since derivatives with respect to $p$ and each $x_m$ commute, it is clear that any collection of the derivatives ${\displaystyle \frac{\partial^i \phi_j}{\partial p^i}}$ will satisfy \eqref{eq:hir} and hence generate solutions to \eqref{eq:KPtau}.  

As a consequence of Lemma \ref{lem:explem}, $p$-derivatives of the $\phi_j$ provide a family of generators for the Wronskian formulation in \eqref{eq:tauWronsk} of solutions $\tau$ to the KP Equation.  Thus to generalize the treatment in \cite{KW-1, KW-2}, it will be useful to first express closed formulas for these expressions. To accomplish this, and provide a library of examples for the remainder of the paper, we calculate the first few $p$-derivatives of $\phi = e^f$.  We shall see that even the simple cases appearing in Table \ref{tab:expder} are suggestive of underlying combinatorial structure in the solutions to \eqref{eq:KPtau}.

%

\begin{center}
\begin{table}[h]
\begin{tabular}{l | l | l | l }
\hline
   $i$
  & $\frac{\p ^i \phi}{\p p^i}$ & partitions of $i$ & derivative coefficients\\
\hline
1 &  $e^f \left( f^\prime\right)$  & $\{1\}$ & $(1)$ \\
2 &  $e^f \left(f^{\prime\prime} + (f^\prime)^2 \right)$ & $\{2\}, \{1, 1\}$ & $(1, 1)$ \\
3 & $e^f \left( f^{\prime\prime\prime} + 3f^\prime f^{\prime\prime} + (f^\prime)^3 \right)$ & $\{3\}, \{1, 2\}, \{1, 1, 1\}$ & $(1, 3, 1)$ \\
4 & $e^f \left( f^{(iv)} + 3(f^{\prime\prime})^2 + 4f^\prime f^{\prime\prime\prime} \right.$ & $\{4\}, \{2, 2\}, \{1, 3\},$ &  $(1, 3, 4, 6, 1)$ \\
 & $\quad \left. + 6(f^\prime)^2 f^{\prime\prime} + (f^\prime)^4 \right)$ & $\quad  \{1, 1, 2\}, \{1, 1, 1, 1\}$  & \\
$\cdots$ & $\cdots$ & $\cdots$ & $\cdots$ \\
\hline
\end{tabular}
\caption{Derivatives of $\phi(p,\bx)=e^{f(p,\bx)}$ with respect to $p$}
\label{tab:expder}
\end{table}
\end{center}

\subsection{Combinatorial formulation of $p$-derivatives}

In this section, our task is to understand the structure of the $p$-derivatives appearing in Table \ref{tab:expder}.  Rather than being concerned with analytical properties of the $\phi(p, \bx)$ appearing in \eqref{eq:tauWronsk}, we instead search for explicit formulas or descriptions of the relevant functions from a combinatorial perspective.  This viewpoint allows us to see patterns and extract information hidden by the differential formulae.  We recall some combinatorial concepts in order to do so.

Each term in the expansion of $\frac {\p ^i \phi}{\p p^i}$ corresponds to a partition, i.e., a way of writing a natural number as a sum.  In general, we will use $\lambda$ to denote a partition, $\ell$ to denote its sum, and $\{a_1, \ldots a_k\}$ to be the terms in the partition.  The symbol $\vdash$ indicates "is a partition of."  By convention, we always arrange the parts of the partition in weakly increasing order, i.e., $a_i \leq a_{i+1}$.  To denote the number of the smallest, second smallest, third smallest, etc., term in a given partition, we use $\beta_1 , \beta_2 , \ldots , \beta_k$.  Note that the number of $\beta$'s and the number of $a$'s are the same, namely $k$, and that we allow for $\beta_j=0$ if necessary.  Further background may be found in \cite{Stanley}.
\begin{eg}
$\{1,1,2\} \vdash 4$, since 1+1+2=4.  In this case, we say that $\{1,1,2\}$ is a partition of 4 which has length 3, with $a_1=1,a_2=1,$ and $a_3=2$ (as repeated numbers are allowed).  Consequently, $\beta_1 = 2, \beta_2 = 1,$ and $\beta_3=0$.

%
With this notation in mind, we return to the $p$-derivatives of $\phi$ listed in Table \ref{tab:expder} and examine $\frac{\partial^4 \phi}{\partial p^4}$.  The exponents and indices of differentiation for each term in the sum correspond to a partition of 4: $f^{(iv)}$ gives the partition $\{4\}$, $(f'')^2$ gives the partition $\{2,2\}$, $f'f'''$ gives the partition $\{1,3\}$, $(f')^2f''$ gives the partition $\{1,1,2\}$, and $(f')^4$ gives $\{1,1,1,1\}$.  In general, the summands of the $k^{th}$ derivative correspond to partitions of $k$.  
\end{eg} 
\begin{eg}
Let $\la=\{2,2,2,2,4,5,5 \}$ so that $\ell = 22$ and $k=7$.  For this partition, we have
$a_1=a_2=a_3=a_4=2$,  $a_5=4$, $a_6=a_7=5$, $\beta_1=4$, $\beta_2=1$, $\beta_3 = 2$ and $\beta_4=\beta_5=\beta_6=\beta_7=0$.  This expression is one of the terms of the partition, indicated in row 22 of Table \ref{tab:expder}, corresponding to the product $(f^{\prime\prime})^4(f^{(iv)})(f^{(v)})$.  
%
\end{eg}

Guided by the previous example, we let $f^{\lambda}$ denote the product of derivatives of $f$ given by $\lambda$.  For example,  if $\lambda=\{1,1,2\}$, then $f^{\lambda}=f'f'f''=(f')^2f''$.  The pattern appearing in the powers of derivatives as well as their coefficients is, in fact, a special case of a classical theorem giving the derivative in terms of combinatorial partitions.

\begin{thm}[Fa\`a di Bruno's Formula \cite{FaadiBruno}]\label{thm:FdB}
If $g$ and $f$ are $\ell^{th}$-order differentiable functions, then
$$\frac{d^\ell}{dt^\ell}[g(f(t))]=\sum_{\la \vdash \ell}\binom{\ell}{a_1, \ldots , a_k}\frac{1}{\beta_1!\cdots \beta_k!}g^{(k)}(f(t))\cdot f^{\la}(t),$$
where ${\displaystyle \binom{\ell}{a_1, \ldots , a_k} = \frac{\ell !}{a_1 ! \cdots a_k !}}$ is a multinomial coefficient.
\end{thm}

This old and straightforward formula nevertheless leads to a broad array of results.  Modern reformulations of Fa\'a di Bruno appear in \cite{KodamaPierce, CK-1, CK-2}, for instance, where relationships between derivatives and partitions are exploited in order to relate Catalan and generalized Catalan numbers to the solutions of the KP hierarchy, by means of chord diagrams, ribbon graphs and other structures.  

We shall use this formula to calculate the arbitrary $p$-derivatives of $\phi(p,\bx) = e^{f(p,\bx)}$ appearing in Lemma \ref{lem:explem}.   Setting $t=p$, $g(t)=e^t$ and $f(t)=tx_1+\cdots +t^Nx_N$, it remains to compute $f^{\la}(t)$ for arbitrary $\la$.  Fortunately, that is fairly simple.  For fixed $\bx$, a direct calculation shows that for $a = 0 , 1, 2,  \ldots$ ,
$$
\frac{\p^a f}{\p p^a}=\sum_{j=1}^{N}(j)_{a}p^{j-a}x_{j},
$$
where $(j)_{a}=j(j-1)(j-2)\ldots(j-a+1)$ is the falling factorial\footnote{Note that $(j)_{a}=0$ for $a>j$ and $(j)_0=1$ for all $j$.}.  Therefore
$$
f^{\la}= \prod_{i=1}^{k}\frac{d^{a_i} f }{dp^{a_i} }= \prod_{i=1}^{k}\left(\sum_{j=1}^{N}(j)_{a_{i}}p^{j-a_{i}}x_{j}\right).
$$
Combining these observations with Fa\`a di Bruno's Formula gives us a version of Theorem \ref{thm:FdB} which provides the input data for our Wronskian solutions \eqref{eq:tauWronsk} in convenient form.  This formulation will allow us to consider a  broader class of functions than those appearing \cite{KW-1, KW-2}, and indicate an expansion to their approach in subsequent sections.
\begin{cor}\label{cor:bigone}
The $\ell^{th}$ derivative of $\phi(p,\bx)=e^{px_{1}+\cdots+p^{N}x_{N}}$ is
\beq\label{eq:bigone}
\frac{d^{\ell}}{dp^{\ell}}\left[e^{px_{1}+\cdots+p^{N}x_{N}}\right]=e^{px_{1}+\cdots+p^{N}x_{N}}\sum_{\la\vdash \ell} \binom{\ell}{a_{1}, \ldots , a_{k}}\frac{1}{\beta_1!\ldots \beta_k!}\prod_{i=1}^{k}\left(\sum_{j=1}^{N}(j)_{a_{i}}p^{j-a_{i}}x_{j}\right).
\eeq
\end{cor}

\section{Generalized KP solitons}

Classical integrable nonlinear PDEs, and more specifically the KP Equation, possess solitary wave solutions.  A natural object of study in the setting of \eqref{eq:KPu} is the family of $N$-solitons or line-solitons.  These are solitary wave solutions which localize into crests or troughs along a finite number of directions and decay exponentially elsewhere.   It is well-known that restricting the generators $\phi_j$ in \eqref{eq:KPtau} to exponentials in $\bx =(x_1, \ldots , x_N)$ will produce solutions of this type.  We shall refer to the associated functions $\tau$  (and consequently $u$ in \eqref{eq:KPu}) as soliton solutions of the equation.  For a complete discussion, see \cite{CK-2, KW-3, Hirota, MJD}, for instance.

Now that we have established formulas for the $p$-derivatives of the exponentials $\phi_j(p,\bx)$, we wish to use them as "generating seeds" in the construction of Wronskian solutions to \eqref{eq:KPtau}.  Although it is possible at this point to substitute these derivatives directly into \eqref{eq:tauWronsk}, we are instead following the approach of \cite{KW-1}, generating a larger family of solutions for combinatorial study.  In particular, we use linear combinations of the $p$-derivatives appearing in Corollary \ref{cor:bigone} as generators to construct the Wronskian.  For ease of notation, set
\beq\label{eq:Fq}
F_q= \frac{\p^{\ell_q}}{\p p_q^{\ell_q}}\phi(p_q, \bx)= e^{p_q x_{1}+\cdots+p_{q}^{N}x_{N}}\sum_{\la\vdash [\ell_{q}]} \binom{\ell_{q}}{a_{1}, \ldots , a_{k}}\frac{1}{\beta_1!\ldots \beta_k!}\prod_{i=1}^{k}\left(\sum_{j=1}^{N}(j)_{a_{i}}p_{q}^{j-a_{i}}x_{j}\right).
\eeq
Note that we are free to choose distinct values of $p$ and $\ell$ for each seed function (each $p$-derivative), which we have indicated by the indexing $p_q$ and $\ell_q$, respectively.  In Lemma \ref{lem:explem}, we showed that such $F_q$ can be used to form Wronskian solutions to \eqref{eq:KPtau}.  By \eqref{eq:hir}, their linear combinations clearly do so as well.  We state this precisely below and include a brief proof, for completeness.
\begin{thm}\label{thm:tauA}
Fix $s, n \in\mathbb{N}$ and let $a_{iq}$ be arbitrary real numbers, $i=1, \ldots , s$, $q=1, \ldots , n$.   Define ${\displaystyle g_i=\sum_{q=1}^{n} a_{iq}F_q}$.  Then the Wronskian $\tau=Wr(g_1, \ldots g_s)$ is  a solution to the KP Equation \eqref{eq:KPtau}.
\end{thm}
\begin{proof}
It is enough to demonstrate that each $g_i$ satisfies the conditions of (\ref{eq:hir}).  Observe that
$$
\frac{\partial g_i}{\partial x_m}=\frac{\partial}{\partial x_m} \left[\sum_{q=1}^{n}a_{iq}F_q\right] = 
\sum_{q=1}^{n}a_{iq}\frac{\partial F_q}{\partial x_m} =  \sum_{q=1}^{n}a_{iq}\frac{\partial^{m}F_{q}}{\partial x_1^{m}} = \frac{\partial^{m}}{\partial x_1^m}\sum_{q=1}^{n}a_{iq}F_q = \frac{\partial^m g_i}{\partial x_1^m},
$$
since each $F_q$ satisfies (\ref{eq:hir}); thus, the same condition holds for $g_i$.
\end{proof}
With these expressions in hand, we can then study their accompanying combinatorial structures.
\begin{rem}
This result may also be proved by viewing the Wronskian in terms of Schur polynomials.  For further discussion from this perspective, see \cite{Sato, MJD}.  As our goal is to study specific combinatorial properties of this particular class of solutions $Wr(g_1, \ldots , g_s)$, we do not need the full generality of that approach.  Instead, we follow the example of \cite{KW-1, KW-2} by considering the $g_i$ as products of matrices and studying the Grassmann structure therein.
\end{rem}

Let $A$ denote the $s \times n$ matrix  having entries $a_{iq}$, $i = 1, \ldots , s$, $q =1, \ldots, n$.  Then 
\beq\label{eq:notation1}
(g_1 \, g_2 \, \cdots \, g_s)^{T}=A\cdot(F_1\, \cdots \, F_n)^{T}.
\eeq
Observe that the vector $(g_1, \ldots , g_s)^T$ is constructed from the set of generators $(F_1, \ldots , F_n)^T$ by way of the full rank matrix $A$ in the Grassmannian $Gr(s,n)$, consisting of $s$-dimensional subspaces of $\mathbb{R}^n$.  Consequently, each solution is described by $A \in Gr(s,n)$ and the collection of constants $p_q$ appearing in $F_1, \ldots , F_n$.

To explore the Wronskian $Wr(g_1, g_2, \ldots , g_s)$ in terms of this matrix $A$, we invoke the Binet-Cauchy identity, which gives the determinant of the product of two rectangular matrices in terms of determinants of their minors:
\beq\label{eq:BinetCauchy}
\det(M'M'')=\sum_{I \in \binom{[t]}{s}}\det(M'_{[s],I})\det(M''_{I,[s]}).
\eeq
Here, $M^\prime$ is an $s \times t$ matrix, $M^{\prime\prime}$ is a $t \times s$ matrix, $[t]=\{1, 2, \ldots , t\}$, $I \in \binom{[t]}{s}$ denotes a subset of $[t]$ of size $s$, and $M_{S,T}$ denotes the minor of a matrix $M$ consisting of rows and columns indexed by the sets $S$ and $T$, respectively.  Thus, in the identity above, $S\subseteq [s]$ and $T \subseteq [t]$.

\begin{eg}
Given a matrix
$$
M=\begin{pmatrix}
1 &2&3 \\ 4&5&6\\7&8&9&\\10&11&12\end{pmatrix},
$$
and sets $S=\{2,4\}$, $T=\{1,3\}$, $M_{S,T}$ is
$$
M_{S,T}=\begin{pmatrix} 4&6\\10&12\end{pmatrix}.
$$
\end{eg}
For the general case being considered in Theorem \ref{thm:tauA}, $M'$ is the $s \times n$ matrix $A = (a_{iq})$ and $M''$ consists of the first $s$ rows of the matrix considered in the Wronskian $Wr(F_1, \ldots F_n)$.  Note that this product is indeed the matrix used for $Wr(g_1, \ldots g_s)$.   Thus, we write 
\beq\label{eq:notation2}
\tau_A=\sum_{I \in \binom{[n]}{s}} \det(A_{[s], I}) Wr(F_{i})_{i \in I}.
\eeq
\begin{rem}\label{rem:KWprod}
In \cite{KW-1, KW-2} Wronskian line-soliton solutions of the KP Equation are generated using seed functions $\{E_{i}\}$ given by ${E_{i}=\exp{(\sum_{j=1}^{m}\kappa_{i}^{j}t_{j})}}$.  It is indicated in \cite{KW-1} that when the $t_{j}$ are treated as constants ($j>3$), the function 
 $$
 \tau_A=Wr(f_1, \ldots , f_k)=\sum_{I \in \binom{[n]}{k}} \Delta_{I}(A)E_{I}
 $$
 is a soliton solution of the KP equation and consequently, only functions of the form $e^{cx+c^2y+c^3t}$ are considered.  Following the notation in \eqref{eq:notation1}, \eqref{eq:BinetCauchy} and \eqref{eq:notation2}, the index $k$ above is our $s$, $\Delta_{I}(A) = \det(A_{[s],I})$ is the determinant of the $k\times k$ submatrix of $A$ given by the columns indexed by $I$ and $E_{I}$ is $Wr((E_{j})_{j \in I})$. 

A combinatorial reason for the restriction in \cite{KW-1} is that for the restricted exponential functions $E_i(x,y,t)$, the derivatives used in $Wr(E_i)_{i\in I}$ give a Vandermonde matrix, whose determinant has the formulation $\prod (\kappa_i-\kappa_j)$.  It is this product formula which allows for an convenient approximation of $\tau_A$ in terms of matrix minors in \eqref{eq:notation2}.  Furthermore, the matrix $A$, viewed as a Grassmannian, gives insight into the structure of $\tau_{A}$ in the sense that local approximations of $\tau_A$ and locations of maxima are ``read off" from $A$'s combinatorial properties.  In \cite{KW-2} this approach is expanded somewhat by allowing for slightly more general exponentials.
\end{rem}

We remark that setting $\ell_j=0$ and relabeling $p_i=\kappa_i$ in \eqref{eq:Fq} recovers all the $\tau$-functions appearing in \cite{KW-1}.  Our treatment generates a broader family of solutions by allowing $\ell_j \neq 0$ in formula \eqref{eq:notation2}, as indicated in Examples section; see Example \ref{eg:example7}.  Consequently, we propose using the formulation \eqref{eq:notation2} to understand more complex $\tau_A$ by an analogous approach to \cite{KW-1, KW-2}.  It would be of particular interest to the authors if one could describe the transformation space or identify an appropriate subspace of a Grassmannian for this family of $N$-solitons to remain invariant in the generalized setting described in this communication.  This seems feasible, as the structure continues to depend on a $s\times n$ matrix A.

\subsection{Combinatorial formulation of $p$-derivatives, revisited}

The first step towards understanding how the combinatorial structure of $\tau_A$ in \eqref{eq:notation2} is informed by and influences soliton dynamics, is accomplished by studying the effects of different choices for $A$, $N$, $\{p_q\}$, and $\{\ell_{q}\}$ on the solution.  The behavior of the $\tau$ function is determined by both the $F_q$ terms and the matrix A.  We shall see that by devising a general recursion formula for the $p$-derivatives of the generating seeds, a product or Vandermonde formula similar to that appearing Remark \ref{rem:KWprod} can be recovered (as in \cite{KW-1, KW-2}).  We motivate such a recursion with a simple example.

\begin{eg}\label{eg:example5}
If $A=\begin{pmatrix} 1&3&5\\2&4&6 \end{pmatrix}$ and we are given $F_1, F_2, F_3$, then
\begin{eqnarray}
  \tau_{A} &=&  \sum_{I \in \binom{[3]}{2}}det(A_{[2],I})Wr(F_i)_{i \in I}  = \sum_{I \in \binom{[3]}{2}}\Delta_{I}(A)F_{I} \nonumber \\ 
  &=& \left| \begin{array}{cc} 1 & 3 \\ 2& 4 \end{array}\right| (F_1F_2'-F_1'F_2)+\left| \begin{array}{cc} 1 & 5 \\ 2& 6 \end{array}\right| (F_1F_3'-F_1'F_3)+\left| \begin{array}{cc} 3 & 5 \\ 4& 6 \end{array}\right| (F_2F_3'-F_2'F_3) \label{eq:eg5tau}
\end{eqnarray}
\end{eg}
Notice $F_q$ and its $x_1$-derivatives appear in the above expression.  To understand a function such as the one in the example, it would be helpful to have a formula for the derivatives with respect to $x_1$.  We now generate such a formula, as well as a compact recursion relation; we include a detailed proof in order to demonstrate the use of combinatorial properties of partitions in such arguments.

\begin{lem}\label{lem:lemma2}
Let $F^{(\ell)}$ be the $\ell^{th}$ derivative with respect to $p$ of $e^{f(p, \bf{x})}=e^{px_1+\cdots +p^N x_{N}}$, as in \eqref{eq:Fq}:  
$$
F^{(\ell)}=e^{p x_{1}+\cdots+p^{N}x_{N}}\sum_{\la\vdash \ell} \binom{\ell}{a_{1}, \ldots , a_{k}}\prod_{i=1}^{k}\left[\left(\sum_{j=1}^{N}(j)_{a_{i}}p^{j-a_{i}}x_{j}\right)/\beta_{i}!\right].
$$
 Then the derivative with respect of $x_1$ of $F^{(\ell)}$ is given by 
 $$\frac{\partial}{\partial x_1}\left[F^ {(\ell)}\right]=pF^{(\ell)}+e^{px_1+\cdots+p^N x_N}\sum_{\la' \vdash\ell-1 }\binom{\ell-1}{a_1', \ldots , a_r'}\prod_{i=1}^{k'}\left[ \left(\sum_{j=1}^{N}(j)_{a_{i}'}p^{j-a_i '}x_j\right)/\beta_i ' !\right].$$  In simpler notation, $$\frac{\partial}{\partial x_1}\left[F^ {(\ell)}\right]=pF^{(\ell)}+\ell F^ {(\ell -1)}.$$
\end{lem}
\begin{proof}
The proof is an explicit computation, using the structure of partitions.  First, observe
  $$\frac{\partial}{\partial x_1}[F^ {(\ell)}]=pF^{(\ell)}+e^{px_1+\cdots+p^N x_N} \frac{\partial}{\partial x_1}\left[ \sum_{\la \vdash \ell}\binom{\ell}{a_1, \ldots , a_k}\prod_{i=1}^{k}\left[\left(\sum_{j=1}^{N}(j)_{a_i}p^{j-a_i}x_j\right)/\beta_i ! \right] \right].$$
 Notice that, by the definition of the falling factorial, $(1)_{a_i}=0$ for $a_i>1$.  Therefore most terms of $F^{(\ell)}$ do not depend on $x_1$.  In fact, the only terms containing $x_1$ are the initial factor of $e^{px_1+p^2x_2+\cdots +p^{N}x_{N}}$ and terms corresponding to partitions $\la$ whose least entry $a_1$ has size one.  In light of this fact, when computing the derivative we need only consider $\la$ such that $a_1=1$:
$$
\frac{\partial}{\partial x_1}[F^{(\ell)}]=pF^{(\ell)}+e^{px_1+\cdots+p^N x_N}\frac{\partial}{\partial x_1}\left[\sum_{\la\vdash \ell, a_1=1}\binom{\ell}{a_1, \ldots , a_k}\prod_{i=1}^{k}\left[\left(\sum_{j=1}^{N}(j)_{a_{i}}p^{j-a_i}x_j\right)/\beta_i ! \right] \right].
$$
We now analyze this formula in more detail, by focusing on 
$$
D=\frac{\partial}{\partial x_1}\left[\sum_{\la\vdash \ell, a_1=1}\binom{\ell}{a_1, \ldots , a_k}\prod_{i=1}^{k}\left[\left(\sum_{j=1}^{N}(j)_{a_{i}}p^{j-a_i}x_j\right)/\beta_i ! \right] \right].
$$
Splitting the product $\prod_{i=1}^{k}$ into two parts, $\{a_i \vert a_i=1\}$ and $\{a_i \vert a_i >1\}$, we rewrite $D$ as 
$$
D=\sum_{\la\vdash \ell, a_1=1}\binom{\ell}{a_1, \ldots , a_k}\frac{\partial}{\partial x_1}\left[\Bigg{\{}\prod_{i=1}^{\beta_1}\left(\sum_{j=1}^{N}(j)_1p^{j-1}x_j\right)^{\beta_1}/\beta_1 !\Bigg{\}}\Bigg{\{}\prod_{i=\beta_1+1}^{k}\left[ \left(\sum_{j=1}^{N}(j)_{a_{i}}p^{j-a_i}x_j \right)/\beta_i!\right] \Bigg{\}}\right].
$$
Since the product $\prod_{i=\beta_1+1}^{k}$ is independent of $x_1$, we need only compute the derivative of the product $\prod_{i=1}^{\beta_1}$.  We have
 \begin{eqnarray*}
 D&=&\sum_{\la\vdash \ell, a_1=1}\binom{\ell}{a_1, \ldots , a_k}\Bigg{\{}\prod_{i=\beta_1+1}^{k}\left(\sum_{j=2}^{N}p^{j-a_i}x_j\right)/\beta_i !\Bigg{\}}\frac{\partial}{\partial x_1}\left[\left(\sum_{j=1}^{N} jp^{j-1} x_j\right)^{\beta_1}/\beta_1 !\right]
  \\&=&\sum_{\la\vdash \ell, a_1=1}\binom{\ell}{a_1, \ldots , a_k}\Bigg{\{}\prod_{i=\beta_1+1}^{k}\left(\sum_{j=2}^{N}p^{j-a_i}x_j\right)/\beta_i !\Bigg{\}} (1/\beta_1 !)(\beta_1)\left(\sum_{j=1}^{N} jp^{j-1}x_j\right)^{\beta_1  -1}
   \\&=&\sum_{\la\vdash \ell, a_1=1}\binom{\ell}{a_1, \ldots , a_k}\Bigg{\{}\prod_{i=\beta_1+1}^{k}\left(\sum_{j=2}^{N}p^{j-a_i}x_j\right)/\beta_i !\Bigg{\}}\left(\sum_{j=1}^{N} jp^{j-1}x_j\right)^{\beta_1 -1}/(\beta_1 -1)!
 \end{eqnarray*}
A well-known combinatorial identity states $\binom{\ell}{1, a_2, \ldots , a_k}=\ell \binom{\ell-1}{a_2, \ldots , a_k}$ (this may be seen by viewing partitions of $\ell$ as partitions of $\ell-1$ and of $1$).  We may rewrite $D$ in terms of partitions $\la'=\{a_1 ', \ldots , a_r '\}$ of $\ell-1$:
$$
D=\ell \sum_{\la' \vdash(\ell-1) }\binom{\ell-1}{a_1', \ldots , a_r'}\prod_{i=1}^{r}\left(\sum_{j=1}^{N}(j)_{a_{i}'}p^{j-a_i '}x_j\right)/\beta_i ' !.
$$
Therefore we see that $$\frac{\partial}{\partial x_1}[F^{(\ell)}]=pF^{(\ell)}+e^{f(p, \bf{x})}D=pF^{(\ell)}+\ell F^{(\ell -1)}.$$
\end{proof}
In fact, this lemma may be generalized up to the $\ell$th derivative of $F^{(\ell)}$ with respect to $x_1$.

\begin{prop}\label{recurse}
For all $\ell$ and all $j\leq \ell$,
\beq
\frac{\partial^{j}}{\partial x_1^j}\left[F^{(\ell)}\right]=\sum_{i=0}^{j}\binom{j}{i}p^i \ell^{j-i}F^{(\ell-j+i)}.
\eeq
\end{prop}
\begin{proof}
This is a simple proof by induction.  The case $j=1$ has already been verified above.  A calculation is enough to show the equation is true for $j+1$ if it is true for $j$, again using standard combinatorial identities:
\begin{eqnarray*}
\frac{\partial^{j+1}}{\partial x_1^{j+1}}\left[F^{(\ell)}\right]&=&\frac{\partial}{\partial x_1}\left[\sum_{i=0}^{j}\binom{j}{i}p^i \ell^{j-i}F^{(\ell-j+i)}\right]\\
&=&\sum_{i=0}^{j}\binom{j}{i}p^{i}\ell^{j-i}[pF^{(\ell-j+i)}+\ell F^{(\ell-j+i-1)}]\\
&=&\ell^{j+1}F^{(\ell-j-1)}+\sum_{i=0}^{j}\left[\binom{j}{i}+\binom{j}{i+1}\right]p^{i+1}\ell^{j-i}F^{(\ell-j+i)}\\
&=&\sum_{i=0}^{j+1}\binom{j+1}{i}p^{i}\ell^{j+1-i}F^{(\ell+i-(j+1))}.
\end{eqnarray*}
\end{proof}

\subsection{Examples}
In this section, we use \eqref{eq:notation2} to calculate solutions $\tau_A$ to the KP Equation explicitly, indicating features worthy of further investigation from this perspective.  The examples may also serve as case studies for any new combinatorial structures identified in future work.  Using the dictionary in Remark \ref{rem:KWprod}, such structures can also then be directly compared or contrasted with those appearing in \cite{CK-1, CK-2, KW-1, KW-2}.
\begin{eg}\label{eg:example6}
In the simplest setting, we recover solutions $\tau_A$ of the form appearing in \cite{KW-1, KW-2}.  Let $N=3$ and set $(x_1,x_2,x_3)=(x,y,t)$.  Fix $A$ to be the full-rank $2\times4$ matrix
$$
\left(\begin{array}{cccc}1&0&1&0 \\ 0&1&0&1\end{array}\right).
$$
Choosing $p_1=1$, $p_2=2$, $p_3=4$, $p_4=8$ and $\ell_1=\ell_2=\ell_3=\ell_4=0$,  equation \eqref{eq:Fq} gives us 
$$
F_1=e^{x+y+t}, \quad F_2=e^{2x+4y+8t}, \quad F_3=e^{4x+16y+64t}, \quad F_4=e^{8x+64y+512t}.
$$
Using \eqref{eq:notation2}, we obtain
\begin{eqnarray*}
\tau_A&=&\sum_{I \in \binom{[4]}{2}}det(A_{[2],I})Wr(F_i)_{i \in I}\\
&=&Wr(F_1,F_2)+Wr(F_3,F_4)+Wr(F_1, F_4)-Wr(F_2,F_3)\\
&=&e^{3x+5y+9t}+4e^{12x+80y+576t}+7e^{9x+65y+513t}-2e^{6x+20y+72t}.
\end{eqnarray*}
It is straightforward to verify $\tau_A$ satisfies \eqref{eq:KPtau}, or equivalently that $u_A$ satisfies \eqref{eq:KPu} using the logarithmic transformation \eqref{eq:logtransf}. 
The standard techniques of \cite{CK-1, KW-1} can be used to determine equations of the lines characterizing the soliton peaks, and will not be repeated here; a graph of the solution surface for fixed $t$ appears below.

\begin{figure}[h]
  \resizebox{12pc}{!}{\includegraphics{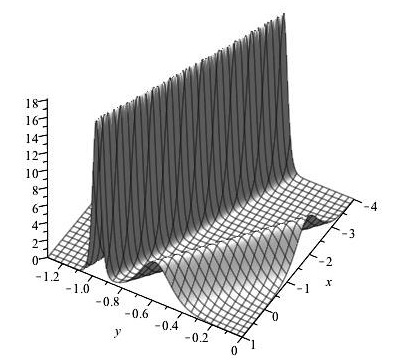}}
\caption{Contour plot of line-soliton solution $u_A(x,y,0.1)$ of KP Equation generated in Example \ref{eg:example6}}
\end{figure}
\end{eg}
As previously mentioned, the generalized formulation in \eqref{eq:Fq} allows for $\ell_j \neq 0$ (in contrast with \cite{KW-1}, for instance, where $\ell_j=0$ for all $j$).  The next example exhibits a second class of solutions whose combinatorial structure has not yet been studied from the perspective appearing in \cite{KW-1, KW-2}, using the totally non-negative Grassmannian, or plabic graph structure.
\begin{eg}\label{eg:example7}
Let $N=3$ and set $A$ to be the matrix in Example \ref{eg:example5}.  Choosing $\ell_1=1$, $\ell_2=1$, and $\ell_3=2$, produces
\begin{eqnarray*}
F_1&=&(x_1+2p_1x_2+3p_1^2x_3)e^{p_1x_1+p_1^2x_2+p_1^3x_3}\\
\frac{\partial F_1}{\partial x_1}&=&(1+p_1x_1+2p_1^2x_2+3p_1^3x_3)e^{p_1x_1+p_1^2x_2+p_1^3x_3}\\
F_2&=&(x_1+2p_2x_2+3p_2^2x_3)e^{p_2x_1+p_2^2x_2+p_2^3x_3}\\
\frac{\partial F_2}{\partial x_1}&=&(1+p_2x_1+2p_2^2x_2+3p_2^3x_3)e^{p_2x_1+p_2^2x_2+p_2^3x_3}\\
F_3&=&(2x_2+6p_3x_3+(x_1+2p_3x_2+3p_3^2x_3)^2)e^{p_3x_1+p_3^2x_2+p_3^3x_3}\\
\frac{\partial F_3}{\partial x_1}&=&(2x_1+6p_3x_2+12p_3^2x_3+p_3(x_1+2p_3x_2+3p_3^2x_3)^2)e^{p_3x_1+p_3^2x_2+p_3^3x_3}.\\
\end{eqnarray*}
These seed functions may be used to produce a generalized solution $\tau_A(x,y,t)$ via \eqref{eq:eg5tau}.  Recall that these generalized solutions are of interest because of their similarity to the functions used in \cite{KW-1, KW-2}.  However, the complexity increases dramatically for these functions.  As a simple indication, observe that for $p_1=1$, $p_2=2$, and $p_3=3$,
$$
\tau_A(x,y,0) = \alpha e^{3x+5y} + \beta e^{4x+10y} + \gamma e^{5x+13y},
$$
where $\alpha = (2x^2+12xy+16y^2-4y)$, $\beta = (112x^2y+480xy^2+32xy+8x^3-16y^2+576y^3-8y+4x^2)$ and $\gamma =  (168xy^2+32x^2y+20xy+288y^3+2x^3+40y^2-4y+2x^2)$.  Thus, a more delicate analysis may be required in order to characterize the dominant exponential regions, necessary for constructing the contour graphs referred to in the previous example.  This will be explored in a subsequent communication.
\end{eg}
Notice that the flexibility of our approach allows for arbitrary $\ell_j$.  They need not be equal nor must they be distinct.  Furthermore, one can still produce solutions to the 2+1-dimensional KP equation using generators $F_q(x_1, \ldots , x_N)$ in  $N$ variables, for $N>3$, by evaluating the resulting $\tau_A(x_1, \ldots , x_N)$ for fixed values of $x_4, \ldots , x_N$.  Such a perspective has not yet been taken into consideration in extending the results of \cite{KW-1, KW-2, CK-1}.

\begin{eg}\label{eg:specialcases}
We now impose a restriction on the generators $F_q$ in \eqref{eq:Fq}.  This reveals connections between $\tau_A$ and several well-known classes of matrices.  Recall that $\tau_A$ is given by a combination of the determinants of the minors of $A$ and Wronskians of the form 
\beq\label{eq:subWronsk}
Wr(F_{j_1}, \ldots , F_{j_k})=\begin{vmatrix} F_{j_1}& F_{j_2}& \cdots& F_{j_k} \\ \frac{\partial F_{j_1}}{\partial x_1}& \frac{\partial F_{j_2}}{\partial x_1}& \cdots& \frac{\partial F_{j_k}}{\partial x_1}\\ \vdots& \vdots& \ddots& \vdots&\\ \frac{\partial^{k-1} F_{j_1}}{\partial x_1^{k-1}}& \frac{\partial^{k-1} F_{j_2}}{\partial x_1^{k-1}}& \cdots& \frac{\partial^{k-1}F_{j_k}}{\partial x_1^{k-1}}&\end{vmatrix}.
\eeq
For simplicity, in this example, we shall suppress the $j$ subscript so that $F_{j_i}$ will be written as $F_i$.  Consider $F_1, \ldots, F_k$ such that $p_i=p$ for all $i$.  In the notation of Lemma \ref{lem:lemma2}, then, $F_i=F^{(\ell_i)}$.  Applying this lemma, we obtain
 \begin{eqnarray*}
 Wr(F_1, \ldots , F_k)&=&\begin{vmatrix} F^{(\ell_1)}& F^{(\ell_2)}& \cdots& F^{(\ell_k)} \\ \frac{\partial F^{(\ell_1)}}{\partial x_1}& \frac{\partial F^{(\ell_2)}}{\partial x_1}& \cdots& \frac{\partial F^{(\ell_k)}}{\partial x_1}\\ \vdots& \vdots& \ddots& \vdots&\\ \frac{\partial^{k-1} F^{(\ell_1)}}{\partial x_1^{k-1}}& \frac{\partial^{k-1} F^{(\ell_2)}}{\partial x_1^{k-1}}& \cdots& \frac{\partial^{k-1}F^{(\ell_k)}}{\partial x_1^{k-1}}&\end{vmatrix}\\
 &=&
 \left|\begin{array}{lll}
 F^{(\ell_1)}& \cdots& F^{(\ell_k)} \\
  pF^{(\ell_1)}+\ell_1F^{(\ell_1-1)}& \cdots& pF^{(\ell_k)}+\ell_kF^{(\ell_k-1)}\\
   \vdots& \ddots& \vdots\\ 
   \displaystyle\sum_{i=0}^{k-1}\binom{k-1}{i}p^i\ell_1^{k-1-i}F^{(\ell_1-k+1+i)}& \cdots & \displaystyle\sum_{i=0}^{k-1}\binom{k-1}{i}p^i\ell_k^{k-1-i}F^{(\ell_k-k+1+i)}\end{array}\right|.
 \end{eqnarray*} 
 In this form, it is clear that elementary row operations can simplify the determinant.  In fact, 
 \beq\label{eq:newWronsk}
 Wr(F_1, \ldots, F_k)=\left|\begin{array}{cccc} F^{(\ell_1)}& F^{(\ell_2)}& \cdots& F^{(\ell_k)} \\ \ell_1F^{(\ell_1-1)}& \ell_2F^{(\ell_2-1)}& \cdots& \ell_kF^{(\ell_k-1)}\\ \vdots& \vdots& \ddots& \vdots\\ \ell_1^{k-1}F^{(\ell_1-k+1)}& \ell_2^{k-1}F^{(\ell_2-k+1)}& \cdots & \ell_k^{k-1}F^{(\ell_k-k+1)}\end{array}\right|.
 \eeq
Observe the duality between this matrix and \eqref{eq:subWronsk}.  The columns of \eqref{eq:newWronsk}  are decreasing $p$-derivatives, whereas in \eqref{eq:subWronsk}, they are increasing $x_1$-derivatives.

The matrix appearing in \eqref{eq:newWronsk} has connections to several well-studied classes of matrices.  It is an example of an alternant matrix.  Furthermore, it can be expressed as a Hadamard or Schur product of the Vandermonde matrix $V=(\ell_{i}^{j-1})_{1\leq i, j \leq k}$ and the matrix of $p$-derivatives $W=(F^{(\ell_i-j+1)})$.  We notice that the Wronskian $E_I$ appearing in \cite{KW-1} can also be written in this manner; in their work, $E_I$ is the Hadamard product of a Vandermonde matrix with entries of the form $\kappa_i^j$, and the matrix with constant columns $(E_i, \ldots, E_i)^{T}$.  For definitions and properties of alternants, Hadamard and Vandermonde matrices, the reader may wish to consult \cite{GenRef1, GenRef2}.
\end{eg}

The structure of the special case in Example \ref{eg:specialcases} indicates the possibility of further relationships between the generalized KP solitons in \eqref{eq:notation2} of this paper and the family of line-solitons of the KP Equation studied in \cite{KW-1, KW-2, CK-1, CK-3}.  Line-solitons admit combinatorial analysis via Grassmannians due to the presence of a Vandermonde determinant (see Remark \ref{rem:KWprod});  the aforementioned duality and matrix factorizations into Vandermonde matrices serve as preliminary evidence for a similar analysis in this setting.  The next step in the analysis will be to find an expression of the determinant that corresponds to some simplifying combinatorial structure in this generating family of functions (just as the product $\prod \kappa_i - \kappa_j$ appears in \cite{KW-1}). 

\section{Summary}
We have summarized several known facts concerning a family of Wronskian solutions to the KP-II equation using a variety of combinatorial methods.  The results are formulated keeping in mind the recent approaches taken by Kodama et. al. in \cite{CK-1, CK-2,CK-3,  KW-1, KW-2, KW-3, KodamaPierce} for this equation.  Generalized KP soliton solutions are expressed in closed form using a family of generators different than those in the recent references in order to suggest possible extensions to their body of work.  As the discussion and Examples indicate, we have some preliminary indication that this family of functions contains enough structure to allow this.

The techniques devised both in this report and in the recent references are applicable to a large class of Wronskian-integrable PDE, for which soliton solutions are known to exist (KdV, Boussinesq, Jimbo-Miwa, etc.).  We suggest that the broad program under consideration involve the study of bilinearizable PDE admitting determinantal solutions; such  equations have been found to possess rich combinatorial structure.  Other concrete connections between integrable PDE and natural combinatorial objects have been successfully studied, e.g.    generating functions in \cite{GouldenJackson}, tests for bilinearizability by way of a modified Bell polynomial \cite{GLNW}.  The current approach sets the stage for further examining such connections.

%
%


\vspace{1cm}

{\bf Acknowledgments. } SB gratefully acknowledges support from Professor M. Kiessling and the NSF through grant DMS-0807705; SB also thanks Professor J. Goldin for several helpful conversations.  AR is supported by the NSF under grant DMS-1004382.  Both authors thank the Department of Mathematics at Rutgers University for providing a stimulating environment while working on this project.

\bibliographystyle{plain}

\end{document}